\documentclass[letterpaper,10pt,conference]{ieeeconf} 

\IEEEoverridecommandlockouts                              

\usepackage{amsmath}
\usepackage{amssymb}

\usepackage{amsthm}
\usepackage{graphicx}
\usepackage{empheq}
\usepackage[caption=false,font=footnotesize]{subfig}
\usepackage{array}
\usepackage{xcolor}
\newcommand{\rev}[1]{{#1}}
\usepackage[normalem]{ulem}
\newif\ifshowdeleted
\showdeletedfalse 
\ifshowdeleted
  \newcommand{\revdelete}[1]{{\color{red}\sout{#1}}}
\else
  \newcommand{\revdelete}[1]{}
\fi

\newtheorem{definition}{Definition}
\newtheorem{assumption}{Assumption}
\newtheorem{theorem}{Theorem}
\newtheorem{lemma}{Lemma}

\newtheorem{corollary}{Corollary}

\newtheorem{problem}{Problem}

\newcommand{\Ag}{\mathcal{I}}          
\newcommand{\Ctr}{\mathcal{C}}         
\newcommand{\Asm}{\mathcal{A}}         
\newcommand{\Gua}{\mathcal{G}}         
\newcommand{\refines}{\preceq}

\title{\LARGE \bf
Contract-Based Decomposition of Temporal Logic Specifications for Networked Systems under Arbitrary Partitions
}

\author{Kodai Kanno$^{1}$, Kenta Hoshino$^{2}$, and Takeshi Hatanaka$^{1}$
\thanks{$^{1}$Kodai Kanno and Takeshi Hatanaka are with the Department of Systems and Control Engineering, Institute of Science Tokyo, Tokyo 152-8550, Japan (e-mail: kanno.k@hfg.sc.e.titech.ac.jp, hatanaka@sc.eng.isct.ac.jp).}%
\thanks{$^{2}$Kenta Hoshino is with the School of Computing, Institute of Science Tokyo, Tokyo 152-8550, Japan, and with DENSO IT Laboratory, Tokyo 105-0004, Japan (e-mail: hoshino@comp.isct.ac.jp).}%
}

\begin{document}
\bstctlcite{IEEEexample:BSTcontrol}

\maketitle
\thispagestyle{empty}
\pagestyle{empty}

\begin{abstract}
  \revdelete{Computational complexity is an inherent limitation of formal synthesis for networked systems, and decomposing the global specification into local ones relaxes this limitation at the cost of conservatism.
Since the granularity of the partition governs this trade-off, it is reasonable to treat the partition as a design variable, which calls for local specifications that remain correct for every partition.
To this end, this letter gives each agent a local specification, written as an assume-guarantee contract that the agent can establish from local information.}
\rev{This letter addresses the problem of decomposing a global signal temporal logic (STL) specification of a networked system, given as a conjunction of local formulas, into local specifications written as assume-guarantee contracts, so that the decomposition is correct for every partition.
The partition can then be changed, without rechecking correctness, to adjust the trade-off between the cost of computing the control inputs against the conservatism caused by the uncertainty that the decomposition introduces.}
\rev{To this end,} we first \rev{present a decomposition that each agent can construct from local information, and} derive a necessary and sufficient condition \revdelete{for these contracts to decompose the global specification}\rev{for its correctness} under a given partition.
\revdelete{Building on this, we}\rev{We} then present a condition under which \revdelete{the decomposition}\rev{it} is correct for every partition\revdelete{, so that the partition becomes a free design variable}.
For linear dynamics and \rev{STL} formulas with affine predicates, we further synthesize a controller \revdelete{for each coalition}\rev{meeting the local specifications} by a tube-based approach.
Finally, simulations on a network of \revdelete{input-coupled }tanks show how the \revdelete{choice of partition trades computational cost against conservatism}\rev{design of the partition balances the trade-off}.
\end{abstract}

\section{INTRODUCTION}
\revdelete{Formal methods provide rich specification languages}\rev{Signal temporal logic (STL) \cite{Maler2004-nj}, a formal language for specifying requirements on the behavior of systems, is a rich language} that \rev{connects} software requirements with physical systems, making it possible to ensure the reliability of cyber-physical systems.
A major challenge in \revdelete{formal synthesis}\rev{synthesizing controllers to meet STL specifications} for networked systems is computational complexity \cite{Lindemann2025-qt}.
For instance, \revdelete{formal synthesis for signal temporal logic (STL)}\rev{the optimal control problem to meet STL} specifications can be cast as a mixed-integer program whose number of binary variables grows with the number of predicates and the horizon length \cite{Belta2019-ve}, which increases the worst-case complexity exponentially and hinders application to large-scale networked systems.

One remedy for the above scalability limitation is to decompose a complex specification of a high-dimensional system into smaller, more tractable specifications on subsystems.
Several works formalize the decomposition of specifications through \revdelete{assume-guarantee contracts}\rev{\emph{assume-guarantee contracts} \cite{Benveniste2018-pp}, which specify a subsystem with an assumption on the behavior of the other subsystems and a guarantee on its own behavior.}\revdelete{, thereby making \revdelete{formal synthesis applicable to}\rev{controller synthesis tractable for} \revdelete{heterogeneous }\rev{large-scale} networked systems. such as power systems \cite{Chen2021-mx}, autonomous robots in a plug-and-play setting \cite{Bodmer2026-xc}, energy systems \cite{Nuzzo2014-mz}, and the design of autonomous driving systems \cite{Incer2025-ac}.}\revdelete{This decomposition, however, increases conservatism, because each controller has less information about the other subsystems and has to guarantee its specification against a correspondingly enlarged uncertainty \cite{Sadraddini2017-am}.}
\rev{This decomposition, however, forces each subsystem to treat the behavior of the others as uncertainty, which introduces conservatism that grows as the partition becomes finer \cite{Sadraddini2017-am}.}
This motivates treating the partition as a design variable that balances conservatism against computational cost.

This design freedom calls for a decomposition that is \emph{correct} for every partition, in the sense that the network meets the global specification whenever every local specification is met.
\rev{Such a decomposition leaves the choice of partition unconstrained by correctness.}
One of the main obstacles is \emph{circularity}, which arises when agents make assumptions about one another and can admit behaviors in which each agent violates its guarantee on the grounds that the other has violated its assumption \cite{Abadi1995-ma}. Sadraddini et al. \cite{Sadraddini2017-am} and Ghasemi et al. \cite{Ghasemi2022-xj} avoid this problem by designing contracts and controllers jointly, forcing each controller to meet the assumptions made by the others.
\revdelete{These frameworks first compute a feasible solution of a mixed-integer program over the grand coalition, which assigns each agent an admissible set that serves as the assumption for its neighbors, and then synthesize local controllers that keep each agent inside it.
Finding such a solution, however, becomes intractable as the network size and the horizon grow.}
\rev{However, these frameworks design the contracts by solving a mixed-integer program over  the behaviors of all subsystems, which reintroduces the computational issue.}
Furthermore, \revdelete{the solution has to be recomputed}\rev{this program has to be solved again} whenever the partition changes.

Another approach to breaking circularity requires implementations to maintain their guarantees one time step beyond the horizon over which their assumptions hold \cite{Abadi1995-ma,Ozay2011-rz}, and Saoud et al. \cite{Saoud2021-md} give an analogous condition for continuous-time systems.
This approach, however, relies on controllers that satisfy this stronger property. 
\revdelete{Systematic constructions of such controllers have been proposed for restricted classes of STL formulas \cite{Liu2025-jb}, but to our knowledge no method covers general specifications expressible in temporal logic, which makes it difficult to treat the partition as a free design variable.}
\rev{Although systematic constructions of such controllers have been proposed for restricted classes of STL formulas \cite{Liu2025-jb}, to our knowledge, no method
covers general specifications expressible in temporal logic, and for such specifications this approach may therefore not be applicable even to a single fixed partition.}

In this letter, \revdelete{we give the specification of each agent as an assume-guarantee contract designed in a distributed manner, and we characterize when such contracts correctly decompose the global specification.}\rev{we give a decomposition of a global specification given as a conjunction of STL formulas into local specifications written as assume-guarantee contracts, and we characterize when such contracts correctly decompose the global specification.}
This characterization is stated solely in terms of the contracts, without assuming any particular \revdelete{controller architecture}\rev{dynamics, predicates, or controller properties}.
\revdelete{First, we derive a necessary and sufficient condition for correctness under a given partition, which is the weakest requirement that the local specifications and the assumptions have to meet.}\rev{First, for a given partition, we show that the condition on the contracts, which is known to be sufficient for correctness \cite{Benveniste2018-pp, Liu2025-jb}, is also necessary.}
Next, \revdelete{building on this result, }we derive a necessary and sufficient condition under which the decomposition is correct for every partition\revdelete{, so that the partition becomes a free design variable.}\rev{.}
Then, for \rev{discrete-time linear systems coupled through the inputs}\revdelete{linear systems} and STL formulas built from affine predicates, we design controllers that satisfy the contract assigned to each \revdelete{coalition}\rev{subsystem}.
Finally, simulations confirm the effectiveness of this design freedom.

\section{PRELIMINARIES}
\subsection{Agents, Coalitions, and Partitions}

Let $\Ag:=\{1,2,\dots,N\}$ be the set of agents.
Agent $i$ has an input $u_i[t]\in\mathbb{R}^{m_i}$ and a state $x_i[t]\in\mathbb{R}^{n_i}$.
The behavior of agent $i$ is the input-state trajectory pair
\begin{equation}
  w_i:=\big(u_i[0],\dots,u_i[T-1],x_i[0],\dots,x_i[T]\big),
  \notag
\end{equation}
where $T$ is a fixed horizon. 
The set of behaviors that agent $i$ can realize is denoted by $\mathbb{M}_i$.
The behavior of the whole network is $w:=(w_1,\dots,w_N)$, and the universe of behaviors is $\mathbb{M}:=\mathbb{M}_1\times\mathbb{M}_2\times\cdots\times\mathbb{M}_N$. We write $w[t]:=\big[u^\top_1[t],\dots,u^\top_N[t],x^\top_1[t],\dots,x^\top_N[t]\big]^\top$.

The agents are dynamically coupled, and this coupling is described by \revdelete{an undirected}\rev{a directed} graph $G=(\Ag,E)$.
The set of neighbors of agent $i$ is denoted by $\mathcal{N}_i:=\{j\in\Ag \mid \revdelete{\{i,j\}}\rev{(j,i)}\in E\}$.
This graph is determined by the plant, while which of the neighboring agents cooperate with one another is a design choice.
\revdelete{A \emph{coalition} \cite{Fele2017-dk} is a group of cooperating agents,}\rev{Each subsystem is expressed as a \emph{coalition} \cite{Fele2017-dk}, which is} formally a subset $C\subseteq\Ag$ inducing a \rev{weakly} connected subgraph of $G$, so that only agents linked by physical coupling can cooperate.
The set of neighbors of a coalition $C$ is $\mathcal{N}_C:=(\bigcup_{i\in C}\mathcal{N}_i)\setminus C$.
A \emph{partition} is a family $P:=\{C_1,C_2,\dots,C_K\}$ of coalitions satisfying $\bigcup_{C\in P}C=\Ag$ and $C_{k_1}\cap C_{k_2}=\emptyset$ for all $k_1\neq k_2$.

\subsection{Assume-Guarantee Contracts}

An \emph{assume-guarantee contract} \cite{Benveniste2018-pp} is a pair $\Ctr=(\Asm,\Gua)$ with $\Asm,\Gua\subseteq \mathbb{M}$.
Any implementation of the contract has the property that each of its behaviors lying in $\Asm$ also lies in $\Gua$.
Every contract is equivalent to its \emph{saturated} form, in which $\Gua$ is replaced by $\Gua\cup(\mathbb{M}\setminus\Asm)$.
Throughout this letter, all contracts are assumed to be saturated.

Let us next introduce the notions of \emph{refinement} and \emph{composition} of contracts.

\begin{definition}[Refinement \cite{Benveniste2018-pp}]
\label{def:refinement}
Let $\Ctr_1=(\Asm_1,\Gua_1)$ and $\Ctr_2=(\Asm_2,\Gua_2)$ be saturated contracts.
$\Ctr_2$ \emph{refines} $\Ctr_1$, written $\Ctr_2\refines\Ctr_1$, if $\Asm_2\supseteq\Asm_1$ and $\Gua_2\subseteq\Gua_1$.
\end{definition}

\begin{definition}[Composition \cite{Benveniste2018-pp}]
Let $\Ctr_l=(\Asm_l,\Gua_l)$, $l\in L$, be a finite family of saturated contracts.
Their \emph{composition} is the contract $\bigotimes_{l\in L}\Ctr_l=(\Asm,\Gua)$ with
\begin{equation}
  \Gua=\bigcap_{l\in L}\Gua_l,\qquad
  \Asm=\Big(\bigcap_{l\in L}\Asm_l\Big)\cup\Big(\mathbb{M}\setminus\bigcap_{l\in L}\Gua_l\Big).
  \label{eq:composition}
\end{equation}
\end{definition}
If a contract $\Ctr_2$ is a refinement of a contract $\Ctr_1$, every implementation of $\Ctr_2$ satisfies $\Ctr_1$. The composition of a family of contracts is the strongest contract satisfied by a system whose components meet them individually.
\subsection{Signal Temporal Logic}
\label{sec:stl}

Signal temporal logic (STL) \cite{Maler2004-nj} specifies properties of discrete-time signals.
It is built from predicates $\mu := \top \text{ (if } p(w[t]) \geq 0 \text{)}, \ \bot \text{ (if } p(w[t]) < 0 \text{)}$, where $p(w[t])\in\mathbb{R}$, according to the recursive grammar
\begin{equation}
  \varphi::=\mu\mid\neg\varphi\mid\varphi_1\wedge\varphi_2\mid \varphi_1\vee\varphi_2\mid 
  F_{[a,b]}\varphi\mid G_{[a,b]}\varphi \mid \varphi_1 U_{[a,b]}\varphi_2.
  \notag
\end{equation}
The symbols $\neg$, $\wedge$, and $\vee$ denote negation, conjunction, and disjunction. The symbols $F$, $G$, and $U$ are the temporal operators ``eventually'', ``always'', and ``until''.
We write $(w,t)\models\varphi$ when the signal $w$ satisfies $\varphi$ at time $t$.

Real-valued predicates allow one to define a robustness degree, which quantifies how strongly a signal satisfies a formula \cite{Donze2010-mp}.
For discrete-time systems, this value can be computed through a \revdelete{a mixed-integer encoding}\rev{\emph{robustness-based encoding}} \cite{Raman2014-rg}.
The encoding consists of inequalities relating a behavior $w$, the predicates, a real-valued variable $\rho^\varphi$, and binary variables $z^\varphi_t$ that indicate whether $\varphi$ is satisfied at time $t$.
We write $\mathcal{C}_\varphi$ for the constraints obtained by recursively applying these encodings.
The largest $\rho^\varphi$ with which $\mathcal{C}_\varphi$ and $z_0^{\varphi}=1$ hold is equal to the robustness degree of $w$ with respect to $\varphi$ \cite{Belta2019-ve}.
This value defines a function $\rho(\varphi,w)$ of the formula and the behavior, and $\rho(\varphi,w)\geq0$ implies $(w,0)\models\varphi$.

Suppose the predicates are affine, that is, $p_l(w[t])=c_l^\top w[t]+d_l$.
Let $w^{\mathrm{nom}}$ be a nominal behavior, and let $w[t]$ satisfy $w[t]-w^{\mathrm{nom}}[t]\in W$ at every time step.
Then the robustness degree of $w$ admits the lower bound \cite{Belta2019-ve}
\begin{equation}
  \rho(\varphi,w)\ \geq\ \rho(\varphi,w^{\mathrm{nom}})
  -\sup_{\delta w\in W}\big\|C_{\varphi} \delta w\big\|_\infty ,
  \label{eq:robust-bound}
\end{equation}
where $C_{\varphi}$ stacks the coefficient vectors $c_l^\top$ of the predicates.

\section{PROBLEM FORMULATION}
\label{sec:prob_formulation}
The dynamics of the network are given by the discrete-time linear model
\begin{equation}
  x[t+1]=Ax[t]+Bu[t],
  \label{eq:network}
\end{equation}
where the state and the input are the concatenations $x[t]:=[x_1^\top[t],\dots,x_N^\top[t]]^\top$, $u[t]:=[u_1^\top[t],\dots,u_N^\top[t]]^\top$ of the states and the inputs of the agents. The inputs are subject to the constraints $u_i[t]\in U_i$ for every $i\in\Ag$ and every $t$, where $U_i$ is convex and contains the origin.
The set $\mathbb{M}_i$ collects the behaviors $w_i$ whose inputs satisfy $u_i[t]\in U_i$ for every $t$ and whose states are generated by \eqref{eq:network} from $(u_i[0],\dots,u_i[T-1])$ and some admissible inputs of the neighbors of $i$.
It is thus a subset of $\prod_{t=0}^{T-1}U_i\times\prod_{t=0}^{T}\mathbb{R}^{n_i}$.
The matrices are the block concatenations
\begin{equation}
  A:=\operatorname{diag}(A_{11},\dots,A_{NN}),
  B:=\begin{bmatrix}
    B_{11} & \cdots & B_{1N}\\
    \vdots & \ddots & \vdots\\
    B_{N1} & \cdots & B_{NN}
  \end{bmatrix}.
  \label{eq:concat-matrix}
\end{equation}
Here $B_{ij}=0$ whenever $j\neq i$ and $j\notin\mathcal{N}_i$.
Since $A$ is block diagonal, the agents are input-coupled, in the sense that an agent is affected by its neighbors only through their inputs.

Partitioning the network into coalitions splits \eqref{eq:network}.
For a coalition $C=\{i_1,\dots,i_l\}$, concatenate $x_C[t]:=[x_{i_1}^\top[t],\dots,x_{i_l}^\top[t]]^\top$, $u_C[t]:=[u_{i_1}^\top[t],\dots,u_{i_l}^\top[t]]^\top$, and $A_C:=\operatorname{diag}(A_{i_1i_1},\dots,A_{i_li_l})$, and, for an agent $j\in \mathcal{N}_C$, define
\begin{equation}
  B_C:=\begin{bmatrix}
    B_{i_1i_1} & \cdots & B_{i_1i_l}\\
    \vdots & \ddots & \vdots\\
    B_{i_li_1} & \cdots & B_{i_li_l}
  \end{bmatrix},
  \quad
  B_{C,j}:=\begin{bmatrix}
    B_{i_1j}\\
    \vdots\\
    B_{i_lj}
  \end{bmatrix}.
  \label{eq:coalition-matrix}
\end{equation}
The dynamics \eqref{eq:network} then split into
\begin{align}
  x_C[t+1] &= A_C x_C[t]+B_C u_C[t]+d_C[t], \label{eq:coalition}\\
  d_C[t] &= \sum_{j\in \mathcal{N}_C}B_{C,j}u_j[t],
  \label{eq:coalition-coupling}
\end{align}
\rev{where $d_C[t]$ aggregates the inputs that the neighbors $j\in\mathcal{N}_C$ outside $C$ apply to the members of $C$.}
The agents in $C$ have no access to the trajectories of the agents outside $C$, so that $d_C[t]$ is unknown to them.
Instead, each agent $i$ holds an assumption $W_j^i\subseteq\prod_{t=0}^{T-1}\mathbb{R}^{m_j}\times\prod_{t=0}^{T}\mathbb{R}^{n_j}$ \rev{for each $j\in\mathcal{N}_i$, which is a set of behaviors of agent $j$ that agent $i$ chooses locally}.
A coalition then plans its trajectories independently, so as to meet its specification for every neighbor behavior admitted by these assumptions.
The only information available to an agent $i\in C$ about a neighbor $j\in\mathcal{N}_C$ is the assumption $W_j^i$, so that the agents in $C$ treat $d_C[t]$ as a disturbance generated by the behaviors of the neighbors admitted by $W_j^i$.

\revdelete{
  The network is required to satisfy a global specification, given as an STL formula $\varphi$ over the behavior of the network.
  We identify the global specification with the set of behaviors $S:=\{w\in\mathbb{M} \mid (w,0)\models\varphi\}$.
  This global specification is decomposed into local formulas $\varphi_i$ assigned to the agents, each written over the behavior of agent $i$ with affine predicates.
  Each $\varphi_i$ induces the set of behaviors $S_i:=\{w_i\in\mathbb{M}_i \mid (w_i,0)\models\varphi_i\}$.
}
\rev{
  Each agent $i$ is assigned a local formula $\varphi_i$ written over its own behavior, and we regard the conjunction $\bigwedge_{i\in\Ag}\varphi_i$ as the global specification.
  Since the agents are affected by their neighbors through the dynamics, we specify each coalition by $\varphi_i$ of its members together with assumptions on the behaviors of its neighbors, as an assume-guarantee contract.
  We identify each $\varphi_i$ with $S_i:=\{w_i\in\mathbb{M}_i \mid (w_i,0)\models\varphi_i\}$ and assume $\emptyset\neq S_i\subsetneq\mathbb{M}_i$ throughout the letter.
  }

A finer partition keeps each coalition small and its trajectories cheap to compute, but leaves more of the neighboring inputs to be treated as a disturbance, so that the local specifications become conservative.
Strengthening the assumptions $W_j^i$ confines that disturbance, yet an assumption that is too strong can be violated by the very neighbors it constrains and can introduce circularity, depending on how the agents are grouped.
We therefore address the problem of giving a necessary and sufficient condition for the local specifications to decompose the global one, which characterizes exactly how far the assumptions $W_j^i$ may be strengthened.

\begin{problem}
\label{prob:decomposition}
Given \revdelete{the STL formula $\varphi$ of the global specification, }the local formulas $\varphi_i$, and the assumptions $W_j^i$, find a necessary and sufficient condition under which the network satisfies \revdelete{$\varphi$}\rev{$\bigwedge_{i\in\Ag}\varphi_i$} whenever every coalition of a given partition meets its local specification, and find a condition under which this holds for every partition.
\end{problem}

We then address the problem of synthesizing a controller for a coalition.

\begin{problem}
\label{prob:controller_synthesis}
Let the dynamics be \eqref{eq:network} and let the predicates of the $\varphi_i$ be affine.
For a coalition $C$, synthesize a controller that meets the local specifications of its members.
\end{problem}
Problems~\ref{prob:decomposition} and \ref{prob:controller_synthesis} are treated in
Sections~\ref{sec:decomposition} and \ref{sec:algorithm}, respectively.

\section{DECOMPOSITION INTO LOCAL SPECIFICATIONS}
\label{sec:decomposition}

In this section, we formalize \emph{correctness} in terms of assume-guarantee contracts and characterize when the local specifications correctly decompose the global specification, first for a given partition and then for every partition. Neither the linear dynamics \eqref{eq:network} nor the affine form of the predicates is assumed here, both being needed only in Section~\ref{sec:algorithm}.

The specification \revdelete{$\varphi$} of the network is written as the assume-guarantee contract \revdelete{$\Ctr:=(\mathbb{M},S)$}\rev{$\Ctr:=(\mathbb{M},\prod_{i\in\Ag}S_i)$}.
Its assumption is the entire universe of behaviors, which means that $\Ctr$ demands the satisfaction of \revdelete{$\varphi$}\rev{$\bigwedge_{i\in\Ag}\varphi_i$} unconditionally.
The local specifications of a coalition $C$ are written, together with its assumptions, as a contract $\Ctr_C$.
We define correctness of the decomposition as follows.

\begin{definition}[Correct decomposition]
\label{def:correctness}
For a partition $P$, the contracts $\{\Ctr_C\}_{C\in P}$ form a correct decomposition of $\Ctr$ if
\begin{equation}
  \bigotimes_{C\in P}\Ctr_C\refines\Ctr.
  \label{eq:correctness}
\end{equation}
\end{definition}

Next, we build the coalition contracts $\Ctr_C$ from $S_i$ and $W_j^i$.

\begin{definition}[Contract of a coalition]
\label{def:coalition}
The members of a coalition $C$ jointly assume that a neighbor $j\in\mathcal{N}_C$ behaves in the set
\begin{equation}
  W_j^C:=\bigcap\{W_j^i \mid i\in C,\ j\in\mathcal{N}_i\},
  \notag
\end{equation}
and the local specification of $C$ is the saturated contract $\Ctr_C:=\big(\Asm_C,\hat{\Gua}_C\big)$ with
\begin{align}
  \Asm_C &:= \{w\in\mathbb{M} \mid w_j\in W_j^C\ \ \forall j\in
  \mathcal{N}_C\},
  \label{eq:coalition-assumption}\\
  \Gua_C &:= \{w\in\mathbb{M} \mid w_C\in S_C\},
  \label{eq:coalition-guarantee}\\
  \hat{\Gua}_C &:= \Gua_C\cup\big(\mathbb{M}\setminus\Asm_C\big),
  \label{eq:coalition-guarantee-sat}
\end{align}
where $w_C:=(w_{i_1},\dots,w_{i_l})$, $i_k\in C$, and $S_C:=\prod_{i\in C}S_i$.
\end{definition}
Note that $W^C_j$ may admit trajectories outside $\mathbb{M}_j$ and can be determined without knowing $\mathbb{M}_j$.
Trajectories in $W^C_j\setminus\mathbb{M}_j$ never occur, and any implementation designed from $W^C_j$ and $\Gua_C$ satisfies the contract $\Ctr_C$.

\rev{
  With the contracts of Definition~\ref{def:coalition}, the property sought in Problem~\ref{prob:decomposition} for a given partition reduces to \eqref{eq:correctness}, as the following lemma shows.
  \begin{lemma}
    \label{lem:correctness}
      For a partition $P$, the following conditions are equivalent.
    (i) The network satisfies $\bigwedge_{i\in\Ag}\varphi_i$ whenever every coalition meets its contract.
    (ii) The refinement \eqref{eq:correctness} holds.
  \end{lemma}
  \begin{proof}
    By \eqref{eq:composition} and Definition~\ref{def:refinement}, the composition $\bigotimes_{C\in P}\Ctr_C$ has guarantee $\Gua_\otimes=\bigcap_{C\in P}\hat{\Gua}_C$ and assumption $\Asm_\otimes=\big(\bigcap_{C\in P}\Asm_C\big)\cup\big(\mathbb{M}\setminus\Gua_\otimes\big)$, and \eqref{eq:correctness} amounts to the two inclusions
  \begin{align}
    &\Asm_\otimes\supseteq \mathbb{M}
    \label{eq:refinement-Asm}\\
    &\Gua_\otimes\subseteq \prod_{i\in\Ag}S_i.
    \label{eq:refinement-Gua}
  \end{align}
  A behavior $w$ meets $\Ctr_C$ when $w \in \Asm_C$ implies $w \in \Gua_C$, that is, $w \in \hat{\Gua}_C$. Thus the behaviors at which every coalition meets its contract form $\Gua_\otimes$, and (i) reads \eqref{eq:refinement-Gua}.
  It therefore suffices to show that \eqref{eq:refinement-Gua} implies \eqref{eq:refinement-Asm}.
  
  By the definition of $\Asm_\otimes$, \eqref{eq:refinement-Asm} holds if and only if
  \begin{align}
    \bigcap_{C\in P}\hat{\Gua}_C\subseteq\bigcap_{C\in P}\Asm_C .
    \label{eq:refinement-Asm2}
  \end{align}
  Suppose that \eqref{eq:refinement-Asm2} fails, and let $w\in\bigcap_{C\in P}\hat{\Gua}_C$ satisfy $w\notin\Asm_{C'}$ for some $C'\in P$.
  If $w\notin\prod_{i\in\Ag}S_i$, then $w$ itself violates \eqref{eq:refinement-Gua}.
  Otherwise, pick $i\in C'$ and $w'_i\in\mathbb{M}_i\setminus S_i$, which exists since $S_i\subsetneq\mathbb{M}_i$, and let $w'$ be $w$ with $w_i$ replaced by $w'_i$.
  Since $\Asm_{C'}$ makes assumptions only on agents outside $C'$, $w'$ still violates $\Asm_{C'}$, so that $w'\in\hat{\Gua}_{C'}$, while $w'_C=w_C\in S_C$ gives $w'\in\Gua_C\subseteq\hat{\Gua}_C$ for every $C\neq C'$.
  Hence $w'\in\Gua_\otimes$.
  However, $w'\notin\prod_{i\in\Ag}S_i$, so \eqref{eq:refinement-Gua} fails.
  \end{proof}
}

Let $\Gamma_P:=(P,E_P)$ be the directed graph with
\begin{equation}
  E_P:=\big\{(D,C) \ \big| \ C,D\in P, \ \exists j\in D\cap\mathcal{N}_C,\
  \mathbb{M}_j\not\subseteq W_j^C\big\},
  \label{eq:assumption-graph}
\end{equation}
so that an edge $(D,C)$ is present whenever a behavior of $D$ can violate the assumption that $C$ makes about $D$.
A necessary and sufficient condition is then given by the following theorem.

\begin{theorem}
\label{thm:refinement}
\revdelete{Let $S_i\neq\emptyset$ for every $i\in\Ag$ and let $P$ be a partition.
Then \feqref{eq:correctness} holds if and only i} \rev{For a partition $P$}, \eqref{eq:correctness} holds if and only if
\begin{align}
  & S_j\subseteq W_j^C, \quad \forall C\in P,\ \forall j\in \mathcal{N}_C,
  \label{eq:cond-cover}\\
  & \Gamma_P \text{ is acyclic}\revdelete{,}\rev{.}
  \label{eq:cond-acyclic}
  \revdelete{\\
  & \textstyle\prod_{i\in\Ag}S_i\subseteq S .
  \label{eq:cond-global}}
\end{align}
\end{theorem}

\begin{proof}
\revdelete{By \eqref{eq:composition}, the composition $\bigotimes_{C\in P}\Ctr_C$ has guarantee $\Gua_\otimes=\bigcap_{C\in P}\hat{\Gua}_C$ and assumption $\Asm_\otimes=\big(\bigcap_{C\in P}\Asm_C\big)\cup\big(\mathbb{M}\setminus\Gua_\otimes\big)$.}
For a behavior $w$, write $P_{\bar{\Gua}}(w):=\{C\in P \mid w_C\notin S_C\}$ and $P_{\bar{\Asm}}(w):=\{C\in P \mid w\notin\Asm_C\}$ for the sets of coalitions whose guarantee and whose assumption, respectively, $w$ violates.
\revdelete{

By Definition~\ref{def:refinement}, \eqref{eq:correctness} amounts to the two inclusions
\begin{align}
&\Asm_\otimes\supseteq \mathbb{M}
\label{eq:refinement-Asm}\\
&\Gua_\otimes\subseteq S.
\label{eq:refinement-Gua}
\end{align}
The first holds if and only if $\bigcap_{C\in P}\hat{\Gua}_C\subseteq\bigcap_{C\in P}\Asm_C$,
 and } \rev{By}\revdelete{by} \eqref{eq:coalition-assumption} and \eqref{eq:coalition-guarantee-sat} the memberships $w\in\bigcap_{C\in P}\hat{\Gua}_C$ and $w\in\bigcap_{C\in P}\Asm_C$ read $P_{\bar{\Gua}}(w)\subseteq P_{\bar{\Asm}}(w)$ and $P_{\bar{\Asm}}(w)=\emptyset$, respectively.
\revdelete{Hence \eqref{eq:refinement-Asm} is equivalent to}\rev{Hence \eqref{eq:refinement-Asm2}, and thus \eqref{eq:refinement-Asm}, is equivalent to}
\begin{align}
  \nexists w\ \text{s.t.}\ P_{\bar{\Gua}}(w)\subseteq P_{\bar{\Asm}}(w)\text{ and }P_{\bar{\Asm}}(w)\neq\emptyset.
  \label{eq:refinement-Asm equivalent}
\end{align}

We first prove that \eqref{eq:cond-cover} and \eqref{eq:cond-acyclic} together imply \eqref{eq:refinement-Asm equivalent}, arguing by contraposition.
Let $w$ satisfy $P_{\bar{\Gua}}(w)\subseteq P_{\bar{\Asm}}(w)$ and $P_{\bar{\Asm}}(w)\neq\emptyset$, and let $C\in P_{\bar{\Asm}}(w)$.
By \eqref{eq:coalition-assumption}, $w_j\in\mathbb{M}_j\setminus W_j^C$ for some $j\in\mathcal{N}_C$, so that $(D,C)\in E_P$ for the coalition $D$ containing $j$.
If \eqref{eq:cond-cover} holds, then $w_j\notin S_j$, and hence $D\in P_{\bar{\Gua}}(w)\subseteq P_{\bar{\Asm}}(w)$.
Every coalition in $P_{\bar{\Asm}}(w)$ therefore receives an edge from a coalition in $P_{\bar{\Asm}}(w)$, and following these edges backward from any of its elements revisits a coalition after finitely many steps.
Thus $\Gamma_P$ contains a cycle, and \eqref{eq:cond-cover} and \eqref{eq:cond-acyclic} cannot both hold.

Conversely, suppose that one of the two fails.
Since $\mathbb{M}$ is a product and every $S_i$ is nonempty, the coordinates of a behavior may be prescribed independently.
If \eqref{eq:cond-cover} fails, choose $C\in P$ and $j\in\mathcal{N}_C$ with $S_j\not\subseteq W_j^C$, and let $w$ satisfy $w_j\in S_j\setminus W_j^C$ and $w_l\in S_l$ for every $l\neq j$.
Then $P_{\bar{\Gua}}(w)=\emptyset$ and $C\in P_{\bar{\Asm}}(w)$, so that \eqref{eq:refinement-Asm equivalent} fails.
If \eqref{eq:cond-acyclic} fails, let $C_1\to C_2\to\cdots\to C_m\to C_{m+1}=C_1$ be a cycle of $\Gamma_P$ with $C_1,\dots,C_m$ pairwise distinct.
By \eqref{eq:assumption-graph}, each edge $(C_s,C_{s+1})$ carries an agent $i_s\in C_s$ with $\mathbb{M}_{i_s}\not\subseteq W_{i_s}^{C_{s+1}}$, and these agents are distinct, the coalitions being disjoint.
Letting $w$ satisfy $w_{i_s}\in\mathbb{M}_{i_s}\setminus W_{i_s}^{C_{s+1}}$ for every $s$ and $w_l\in S_l$ for the remaining agents gives $P_{\bar{\Asm}}(w)\supseteq\{C_1,\dots,C_m\}\supseteq P_{\bar{\Gua}}(w)$, so that \eqref{eq:refinement-Asm equivalent} fails again.
Conditions \eqref{eq:cond-cover} and \eqref{eq:cond-acyclic} are thus together equivalent to \eqref{eq:refinement-Asm}.

It remains to treat \eqref{eq:refinement-Gua}.
Assume \eqref{eq:cond-cover} and \eqref{eq:cond-acyclic}, so that $\Gua_\otimes\subseteq\bigcap_{C\in P}\Asm_C$ by the above.
Every $w\in\Gua_\otimes$ then satisfies $w\in\hat{\Gua}_C\cap\Asm_C=\Gua_C\cap\Asm_C$ for each $C\in P$, whence $w\in\bigcap_{C\in P}\Gua_C=\prod_{i\in\Ag}S_i$\revdelete{, and \eqref{eq:cond-global} places this set in $S$}\rev{, which is \eqref{eq:refinement-Gua}}.
\revdelete{Conversely, \eqref{eq:refinement-Gua} gives $\prod_{i\in\Ag}S_i=\bigcap_{C\in P}\Gua_C\subseteq\Gua_\otimes\subseteq S$, which is \eqref{eq:cond-global}.}
\rev{Conversely, if \eqref{eq:refinement-Gua} holds, then \eqref{eq:refinement-Asm} holds by the proof of Lemma~\ref{lem:correctness}, and hence \eqref{eq:cond-cover} and \eqref{eq:cond-acyclic} hold.}
\end{proof}

Let $P_1:=\{\{1\},\{2\},\dots,\{N\}\}$ be the partition into singletons.
The following theorem reduces correctness for every partition to the singleton case under a stronger condition.

\begin{theorem}
\label{thm:free}
\revdelete{Let $S_i\neq\emptyset$ for every $i\in\Ag$.}
\rev{Let \eqref{eq:cond-cover} hold for $P_1$.}
Then \eqref{eq:correctness} holds for every partition if and only if \revdelete{\eqref{eq:cond-cover} and \eqref{eq:cond-global} hold for $P_1$ and }there do not exist edges $(\{j_1\},\{i_1\}),(\{j_2\},\{i_2\})\in E_{P_1}$ and disjoint coalitions $C$ and $D$ that satisfy
\begin{equation}
  \{i_1,j_2\}\subseteq C,\qquad \{j_1,i_2\}\subseteq D .
  \label{eq:cond-link}
\end{equation}
\end{theorem}

\begin{proof}
Conditions \eqref{eq:cond-cover} \revdelete{and \eqref{eq:cond-global} }hold for every partition if they hold for $P_1$.
It therefore remains to show that \eqref{eq:cond-acyclic} holds for every partition if and only if no pair of edges and no pair of disjoint coalitions satisfies \eqref{eq:cond-link}.

By the definition~\eqref{eq:assumption-graph} of $E_P$, if a partition $P$ contains coalitions $C$ and $D$ as in \eqref{eq:cond-link}, then $E_P$ contains both $(C,D)$ and $(D,C)$, so that $\Gamma_P$ has a cycle.
Conversely, let $\Gamma_P$ contain a cycle $C_1\to\cdots\to C_m\to C_1$.
Consecutive coalitions are \revdelete{adjacent in $G$}\rev{joined by an edge of $G$}, so $C_2\cup\cdots\cup C_m$ is again a coalition and forms a pair of opposite edges with $C_1$.
\end{proof}

The condition of Theorem~\ref{thm:free} ranges over pairs of edges, so that it has to be checked globally.
The following corollary instead imposes a stronger requirement on the assumptions, which the agents can check on their own. 
This requirement cannot be weakened \revdelete{without exploiting the structure of $G$}\rev{even for a sparse graph such as a directed ring},
as the corollary shows.

\begin{corollary}
\label{cor:free}
\revdelete{Let $S_i\neq\emptyset$ for every $i\in\Ag$, and }Let \eqref{eq:cond-cover} \revdelete{and \eqref{eq:cond-global} }hold for $P_1$.
Then 
\eqref{eq:correctness} holds for every partition if at most one singleton has nonzero in-degree in $\Gamma_{P_1}$, or at most one singleton has nonzero out-degree.
This condition is necessary when $G$ is \revdelete{the complete graph}\rev{the directed ring, with $E=\{(i-1,i)\mid i\in\Ag\}$ and indices taken modulo $N$}.
\end{corollary}
\begin{proof}
If there is at most one singleton with nonzero in-degree, $i_1=i_2$ holds for any pair of edges $(\{j_1\},\{i_1\}),(\{j_2\},\{i_2\})$ in $E_{P_1}$.
Therefore no disjoint coalitions can satisfy \eqref{eq:cond-link}, and \eqref{eq:correctness} follows from Theorem~\ref{thm:free}.
\revdelete{The case of a unique singleton with nonzero out-degree is symmetric.}\rev{The out-degree case is symmetric.}

\revdelete{If there are at least two singletons with nonzero in-degree and at least two with nonzero out-degree, then there exist edges $(\{j_1\},\{i_1\}),(\{j_2\},\{i_2\})$ in $E_{P_1}$ with $i_1\neq i_2$ and $j_1\neq j_2$. When $G$ is complete, every subset of $\Ag$ can be a coalition, so the partition $\{\{i_1,j_2\},\{j_1,i_2\}\}\cup\{\{i\}\mid i\in\Ag\setminus\{i_1,j_1,i_2,j_2\}\}$ satisfies \eqref{eq:cond-link}.}
\rev{When $G$ is the directed ring, two distinct edges of $E_{P_1}$ take the form $(\{i_1-1\},\{i_1\})$ and $(\{i_2-1\},\{i_2\})$ with $i_1\neq i_2$, and the coalitions $\{i_1,\dots,i_2-1\}$ and $\{i_2,\dots,i_1-1\}$, indices taken modulo $N$, form a partition satisfying \eqref{eq:cond-link}.}
\end{proof}

\rev{Corollary~\ref{cor:free} allows nontrivial assumptions $\mathbb{M}_j\not\subseteq W_j^i$ to be made only by a single agent or only about a single agent.
In return, once this agent is fixed, the partition can be changed without rechecking correctness and without any requirement on the controllers.
Conservatism can thus be reduced by enlarging the coalitions, whereas \cite{Abadi1995-ma,Ozay2011-rz,Saoud2021-md} reduce it through nontrivial assumptions between all neighbors, at the price of a stronger requirement on every controller.}

\revdelete{The condition of Corollary~\ref{cor:free} holds trivially when $\mathbb{M}_j\subseteq W_j^i$ for every $i\in\Ag$ and every $j\in\mathcal{N}_i$.
Otherwise, once the agents agree on which one of them carries the nonzero in-degree or out-degree, each agent can locally decide how strong the assumption about its neighbors can be.
Corollary~\ref{cor:free} therefore makes the partition a free design variable, to be chosen without further concern for correctness.}

\section{CONTROLLER DESIGN FOR A COALITION}
\label{sec:algorithm}

In this section, following \cite{Belta2019-ve}, we design a controller for a coalition $C$ that meets the contract $\Ctr_C$.
Throughout this section we assume the linear dynamics \eqref{eq:network} and affine predicates in the local formulas $\varphi_i$. In addition, the $\varphi_i$ are taken as given\revdelete{, with $\bigwedge_{i\in\Ag}\varphi_i$ implying $\varphi$}.
The assumptions $W_j^i$ take the form $\prod_{t=0}^{T-1}U_j^i\times\prod_{t=0}^{T}\mathbb{R}^{n_j}$, where $U_j^i$ is the set of inputs that agent $i$ requires of agent $j$.
They are taken to satisfy the hypothesis of Corollary~\ref{cor:free}.

With the dynamics (\ref{eq:coalition}), the disturbance in \eqref{eq:coalition-coupling} that acts on agent $i$ is 
\begin{equation}
  D^C_i:=\Big\{d \ \Big| \ d=\!\!\!\sum_{j\in(\mathcal{N}_i\cap\mathcal{N}_C)}\!\!\!
  B_{ij}u_j,\ u_j\in U_j^i\Big\}.
  \notag
\end{equation}
To plan the trajectories of a coalition against such a set-valued disturbance, we adopt the tube-based approach \cite{Mayne2005-jv}.
Its construction rests on a stabilizing feedback and an invariant set for each agent.

\begin{assumption}
\label{as:tube}
For every agent $i\in\Ag$ there exist a gain $K_i$ and a bounded convex set $Z_i\subseteq\mathbb{R}^{n_i}$ satisfying $(A_{ii}+B_{ii}K_i)Z_i\oplus D^{\{i\}}_i\subseteq Z_i$ and $K_iZ_i\subseteq U_i$, where $0\in Z_i$ and $\oplus$ denotes the Minkowski sum.
\end{assumption}

\rev{That is, each agent~$i$ can confine its deviation from the plan to $Z_i$ under the input constraint $U_i$ against $D_i^{\{i\}}$, which holds if $(A_{ii},B_{ii})$ is stabilizable and $D_i^{\{i\}}$ is small enough.}
Let the input of agent $i$ be $u_i[t]=v_i[t]+u_i^{\mathrm{fb}}[t]$, where $v_i[t]$ and $u_i^{\mathrm{fb}}[t]$ are the feedforward and feedback parts, respectively.
The feedforward $v_i[t]$ generates the nominal behavior $w_i^{\mathrm{nom}}:=(v_i[0],\dots,v_i[T-1],x_i^{\mathrm{nom}}[0],\dots,x_i^{\mathrm{nom}}[T])$ through \eqref{eq:coalition} with $d_C\equiv0$.
The feedback part is taken as
$u_i^{\mathrm{fb}}[t]:=K_i\big(x_i[t]-x_i^{\mathrm{nom}}[t]\big)$.
This feedback confines the deviation from the nominal state to $Z_i$ whenever the disturbance lies in $D^{\{i\}}_i$ at every time step.

Agents that belong to the same coalition as agent $j$ know the feedforward input $v_j[t]$, the gain $K_j$, and the set $Z_j$, so that the feedback $u_j^{\mathrm{fb}}[t]$ is the only part of the input of $j$ that remains uncertain to them.
The deviation of agent $i\in C$ is therefore driven by ${D'}^{C}_i:=D^C_i\oplus\bigoplus_{j\in\mathcal{N}_i\cap C}B_{ij}K_jZ^C_j$
in which $Z_j^C$ denotes the set confining the deviation of the member $j$.
If the sets $Z_i^C$ satisfy
\begin{equation}
  \forall i\in C,\ (A_{ii}+B_{ii}K_i)Z^C_i\oplus {D'}^{C}_i\subseteq Z^C_i,\ K_iZ^C_i\subseteq U_i,
  \label{eq:rpi_coalition}
\end{equation}
the feedback $u_i^{\mathrm{fb}}[t]$ keeps the deviation in $Z^C_i$.
If $K_jZ_j\subseteq U_j^i$ for all \revdelete{adjacent $i,j\in C$}\rev{$i\in C$ and $j\in\mathcal{N}_i\cap C$}, then substituting
$Z_j$ for $Z_j^C$ gives ${D'}^C_i\subseteq D^{\{i\}}_i$, and hence $Z_i$ satisfies \eqref{eq:rpi_coalition}. 
Tighter sets can be then found, for example, among the scalings $Z^C_i:=\gamma_iZ_i$, by decreasing the $\gamma_i\geq 0$ from $1$ as long as \eqref{eq:rpi_coalition} holds.

Then $u_i^{\mathrm{fb}}[t]$ confines the deviation $x_i[t]-x_i^{\mathrm{nom}}[t]$ to $Z_i^C$ at every time step, and the feedback itself lies in $K_iZ_i^C$.
Therefore \eqref{eq:robust-bound} gives the lower bound $\rho(\varphi_i,w_i^{\mathrm{nom}})-\varepsilon_i^C$ on the robustness degree achieved by agent $i$, where $\varepsilon_i^C:=\sup_{\delta w_i \in K_iZ_i^C\times Z_i^C}
\|C_{\varphi_i}\delta w_i\|_\infty$ is a constant.
The members of $C$ have to meet their formulas simultaneously, so the margin of the coalition is the smallest of the individual margins, and we define the robustness degree of the coalition as $\min_{i\in C}(\rho(\varphi_i,w_i^{\mathrm{nom}})-\varepsilon_i^C)$.
The controller of $C$ is then obtained from the mixed-integer program in \cite{Belta2019-ve}, extended to a coalition,
\begin{equation}
\begin{aligned}
  \min\ & \textstyle\sum_{t=0}^{T-1}v_C^\top[t] Rv_C[t]-\lambda\rho_C\\
  \text{s.t.}\ & x_C^{\mathrm{nom}}[t+1]
      =A_Cx_C^{\mathrm{nom}}[t]+B_Cv_C[t],\\
  & v_i[t]\in U_i\ominus K_iZ_i^C,\\
  & \rho_C\leq\rho^{\varphi_i}-\varepsilon_i^C,\\
  & \mathcal{C}_{\varphi_i},\ z_0^{\varphi_i}=1,\\
  & t=0,1,\dots,T-1,\quad i\in C ,
\end{aligned}
\label{eq:miqp}
\end{equation}
where $\ominus$ denotes the Pontryagin difference, and the constraints $\mathcal{C}_{\varphi_i}$, $z_0^{\varphi_i}=1$ encode $\varphi_i$ as in Section~\ref{sec:stl}.
Here $x_C^{\mathrm{nom}}[0]$ is the measured state, $R\succeq0$ weights the control effort, and $\lambda>0$ weights the margin.
\rev{The weight $\lambda$ is taken large enough that the control effort does not prevent the solution from attaining $\rho_C\geq0$.}
\rev{Solving \eqref{eq:miqp} requires the members of $C$ to share their models, local formulas, initial states, gains $K_i$, and sets $Z_i$.}
The real-valued variable $\rho_C$ is equal to $\min_{i\in C}(\rho(\varphi_i,w_i^{\mathrm{nom}})-\varepsilon_i^C)$ at an optimal solution. A plan with $\rho_C\geq0$ meets the local formulas of all the members of $C$ if every neighbor $j\in\mathcal{N}_C$ behaves in $W_j^C$.

\section{SIMULATION}
\label{sec:simulation}
The proposed decomposition and controller are illustrated on a network of input-coupled tanks obtained by extending the quadruple-tank process of \cite{Johansson2000-be}.
In the plant, agent $i\in\Ag=\{1,\dots,12\}$ owns upper tank $i$ and lower tank $i+12$ together with pump $i$, whose flow is split between the agent's own upper tank and the upper tank of another agent, as shown in Fig.~\ref{fig:agent}.
Each upper tank discharges into the lower tank of the same agent, and each lower tank discharges into a reservoir.
The agents are interconnected in the \rev{directed} ring of Fig.~\ref{fig:network}.

\begin{figure}[t]
\centering
\subfloat[\label{fig:agent}]
{\raisebox{\dimexpr25bp-0.5\height\relax}[50bp][0pt]{\includegraphics[scale=0.24]{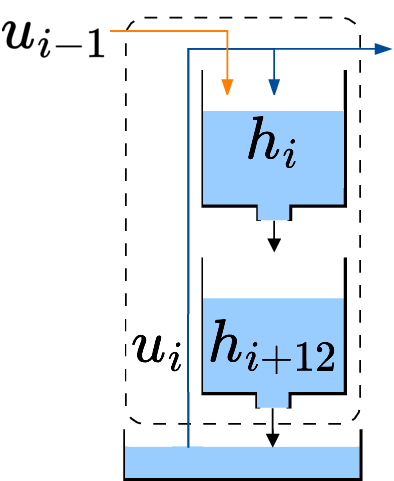}}}
\subfloat[\label{fig:network}]
{\raisebox{\dimexpr25bp-0.5\height\relax}[50bp][0pt]{\includegraphics[scale=0.34]{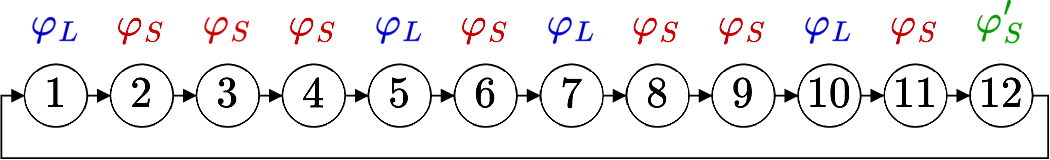}}}
\caption{\rev{Diagrams of (a) an agent and (b) the network of agents.}}
\label{fig:system_diagram}
\end{figure}

\revdelete{
\begin{figure}[t]
\centering
\subfloat[An agent\label{fig:agent}]
{\raisebox{\dimexpr44.8bp-0.5\height\relax}[89.6bp][0pt]{\includegraphics[scale=0.34]{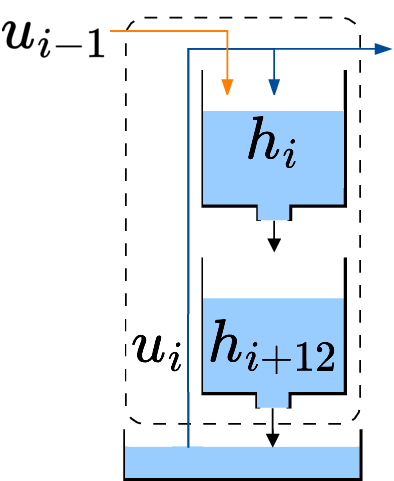}}}
\subfloat[The network of agents\label{fig:network}]
{\raisebox{\dimexpr44.8bp-0.5\height\relax}[89.6bp][0pt]{\includegraphics[scale=0.34]{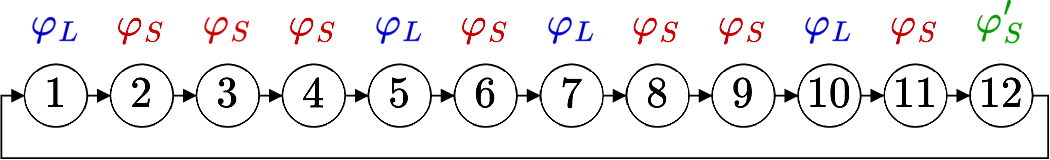}}}
\caption{Diagram of an agent and of the network of agents.}
\label{fig:system_diagram}
\end{figure}
}

The state of agent $i$ is $x_i[t]:=[h_i[t],h_{i+12}[t]]^\top$, where $h_\ell[t]$ denotes the level of tank $\ell$, and the input $u_i[t]$ is the voltage applied to pump $i$.
Both $h_\ell[t]$ and $u_i[t]$ are taken as deviations from an equilibrium.
With the physical parameters of \cite{Johansson2000-be} and a zero-order hold discretization at the sampling period $T_s=4\,\mathrm{s}$, the matrices of \eqref{eq:concat-matrix} are
\begin{align}
  A_{ii}&=\begin{bmatrix}0.938 & 0\\ 0.060 & 0.938\end{bmatrix},\
  B_{ii}=\begin{bmatrix}0.187\\ 0.006\end{bmatrix},\
  B_{i,i-1}=\begin{bmatrix}0.301\\ 0.010\end{bmatrix}\notag
\end{align}
for odd $i$ and
\begin{align}
  A_{ii}&=\begin{bmatrix}0.957 & 0\\ 0.042 & 0.957\end{bmatrix},\
  B_{ii}=\begin{bmatrix}0.137\\ 0.003\end{bmatrix},\
  B_{i,i-1}=\begin{bmatrix}0.219\\ 0.005\end{bmatrix}\notag
\end{align}
for even $i$, the remaining blocks being zero.
\rev{The inputs are constrained to $U_i := [-1,1]$ for $i\neq 12$ and $U_{12} := [-1.5,1.5]$, and the assumptions are $U^i_j := U_j$ except the nontrivial one $U^1_{12} := [-1,1]$.}
\revdelete{The inputs are constrained to $U_i:=[-1,1]$, and the assumptions $U^i_j$ of the agents are set equal to $U_j$.} 
The horizon is $T=75$ steps, that is, $300\,\mathrm{s}$.
\revdelete{Agents $1$, $5$, $7$, and $10$ are assigned the loose formula}
\rev{
  Each agent is assigned, as indicated in Fig.~\ref{fig:network}, the loose formula $\varphi_L$ or one of the strict formulas $\varphi_S$ and $\varphi'_S$, where
\begin{align}
        \varphi_L&:=G_{[0,300]}(-8\leq h_i\leq 8)\wedge F_{[0,200]}G_{[0,100]}(-4\leq h_{i+12}),\notag\\
        \varphi_S&:=G_{[0,300]}(-5\leq h_i\leq 5)\wedge F_{[0,200]}G_{[0,100]}(-1\leq h_{i+12}),\notag\\
        \varphi'_S&:=\varphi_S\wedge G_{[0,300]}(-1\leq u_{i}\leq 1).\notag
\end{align}
}
The controller \eqref{eq:miqp} is evaluated on two kinds of partitions.
We define the uniform partitions $P_k:=\{C_{k,m}\mid m\in\{1,\dots,12/k\}\}$, where $C_{k,m}:=\{(m-1)k+1,(m-1)k+2,\dots,mk\}$.
The tailored partition $P_d:=\{\{1,2,3,4\},\{5,6\},\{7,8,9\},\{10,11,12\}\}$ instead places the coalition boundaries so that every member exposed to a neighbor outside its coalition carries the loose formula.
The simulations were run on a laptop with a $2.0$ GHz Intel Core Ultra 7 processor, using Gurobi as the optimizer.

\begin{table}[t]
\caption{\rev{Performance for each partition}}
\label{table:partition_results}
\centering
\begin{tabular}{|c||c|c|c|>{\centering\arraybackslash}p{2.6cm}|}
\hline
& $\min \rho_C$ & $\max t_{\mathrm{comp}}$ [s]& $\max n_b$ & Failed coalitions \\
\hline\hline
$P_1$ & $-2.18$ & $0.34$ &  $127$ & $\{C_{1,m}\mid m=2,3,4,6,8,9,11,12\}$\\
\hline
$P_4$ & $-1.76$ & $5.28$ &  $508$ & $\{C_{4,m}\mid m=3\}$\\
\hline
$P_6$ & $+1.24$& $27.28$ &  $762$ & \textemdash \\
\hline
$P_{12}$ & $+2.83$ & $161.89$ & $1524$ & \textemdash \\
\hline
$P_d$ & $+0.82$ & $5.44$ &  $508$ & \textemdash \\
\hline
\revdelete{
$P_1$ & $-2.18$ & $0.36$ &  $127$ & $\{C_{1,m}\mid m=2,3,4,6,8,9,11,12\}$\\
\hline
$P_2$ & $-1.76$ & $1.34$ &  $254$ & $\{C_{2,m}\mid m=2,5,6\}$\\
\hline
$P_3$ & $-2.18$ & $7.67$ &  $381$ & $\{C_{3,m}\mid m=2\}$\\
\hline
$P_4$ & $-1.75$ & $5.43$ &  $508$ & $\{C_{4,m}\mid m=3\}$\\
\hline
$P_6$ & $+1.33 $& $18.56$ &  $762$ & \textemdash \\
\hline
$P_{12}$ & $+2.83$ & $113.02$ & $1524$ & \textemdash \\
\hline
$P_d$ & $+1.09$ & $7.17$ &  $508$ & \textemdash \\
\hline
}
\end{tabular}
\end{table}

For each partition, Table~\ref{table:partition_results} shows the minimum robustness degree $\rho_C$, the maximum computation time $t_{\mathrm{comp}}$ of \eqref{eq:miqp}, and the maximum number $n_b$ of binary variables in \eqref{eq:miqp} over the coalitions in the partition, together with the coalitions that fail.
\rev{Fig.~\ref{fig:trajectories} compares the trajectories of agent $1$ under $P_1$ with the trivial assumption $\tilde{U}^1_{12}:=U_{12}$ and with $U^1_{12}$.
Under $\tilde{U}^1_{12}$, the tube of agent $1$ leaves the region admitted by $\varphi_L$, and $C_{1,1}$ fails.
The nontrivial $U^1_{12}$ shrinks $D^{\{1\}}_1$, and with it $Z_1$ and $\varepsilon^{C_{1,1}}_1$, so that $C_{1,1}$ succeeds.
This success results from reduced conservatism, and by Corollary~\ref{cor:free}, this gain does not compromise correctness under any partition.
}

\rev{Next,} as Table~\ref{table:partition_results} shows, the number of binary variables grows linearly with the coalition size, whereas the computation time grows far more rapidly.
Among the uniform partitions, the robustness degree becomes nonnegative only for the coarse ones $P_6$ and $P_{12}$.
A coalition fails when a member carrying the strict formula is exposed to inputs from outside the coalition.
Such exposure forces that member to account for every input its outside neighbor may apply, so that its tube grows and leaves the region admitted by its formula.
\revdelete{Fig.~\ref{fig:trajectories} illustrates this effect.
Agent $3$ is exposed under $P_1$ (Fig.~\ref{fig:traj-uniform}), whereas under $P_d$ (Fig.~\ref{fig:traj-tailored}) its neighbor belongs to the same coalition.}
The tailored partition avoids this effect while keeping the coalitions small, and thereby attains the shortest computation time among the partitions that succeed.

\begin{figure}[t]
\centering
\subfloat[\label{fig:traj-trivial-assumption}]
{\includegraphics[width=0.34\columnwidth]{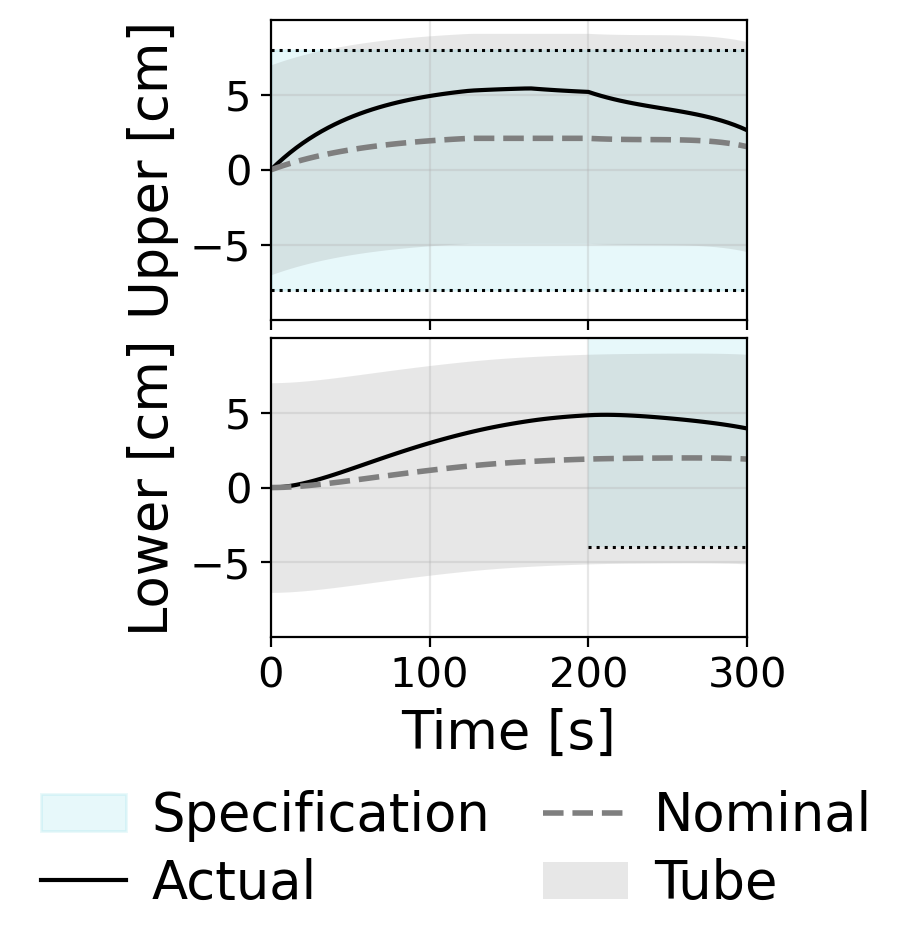}}%
\hspace{0.13\columnwidth}%
\subfloat[\label{fig:traj-nontrivial-assumption}]
{\includegraphics[width=0.34\columnwidth]{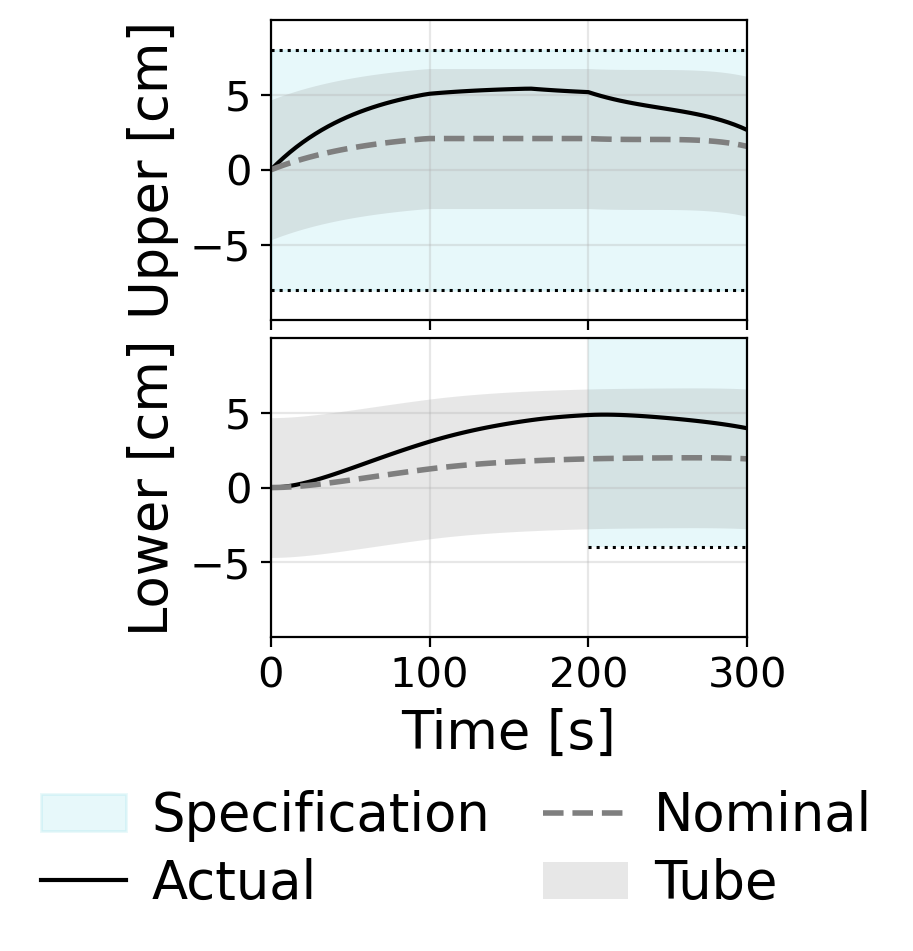}}
\caption{\rev{Trajectories of agent $1$ under $P_1$ with (a) the trivial assumption $\tilde{U}^1_{12}$ and (b) $U^1_{12}$.}}
\label{fig:trajectories}
\end{figure}

\revdelete{
\begin{figure}[t]
\centering
\subfloat[Partition $P_1$\label{fig:traj-uniform}]
{\includegraphics[width=0.45\columnwidth]{figures/contract_stl_ring12_agent3_12x1.png}}%
\hfill
\subfloat[Partition $P_d$\label{fig:traj-tailored}]
{\includegraphics[width=0.45\columnwidth]{figures/contract_stl_ring12_agent3_tailored.png}}
\caption{Trajectories of agent $3$ under the partitions $P_1$ and $P_d$.}
\label{fig:trajectories}
\end{figure}
}

\section{CONCLUSION}
We presented a necessary and sufficient condition for local specifications given as assume-guarantee contracts to correctly decompose a global specification, and characterized when this holds for every partition.
These results turn the partition into a design variable.
For linear systems and STL formulas with affine predicates, we synthesized controllers that satisfy the contract assigned to each coalition.
Simulations confirmed the effectiveness of designing the partition.

Future work includes algorithms that form coalitions adaptively, building on the coalitional control literature \cite{Fele2017-dk}.
Since this framework provides access to the behavior of the other members of a coalition, such a mechanism may also offer a way to resolve the deadlocks that arise when agents lack this access \cite{Oh2019-hp}.

\addtolength{\textheight}{0cm}   





\section*{ACKNOWLEDGMENT}
Claude (Anthropic) \cite{claude} was used solely for language editing and grammar correction during manuscript preparation. All scientific content, interpretations, and conclusions were produced and verified by the authors.


\bibliographystyle{IEEEtran} 
\bibliography{root}   

\end{document}